\documentclass[letterpaper,journal]{IEEEtran}

\usepackage[english]{babel}

\usepackage{newtxtext,newtxmath}

\usepackage{amsmath,amsthm}
\usepackage{amsfonts}
\usepackage{mathrsfs}
\usepackage[hidelinks]{hyperref}

\usepackage{tikz}
\usetikzlibrary{shapes,calc,arrows,patterns,decorations.pathmorphing
	,decorations.markings} \usetikzlibrary{arrows.meta}

\newcommand{\T}{^{\mbox{\tiny T}}}
\newcommand{\R}{\mathbb{R}}

\newcommand{\eps}{\varepsilon}
\let\leq\leqslant
\let\geq\geqslant

\newenvironment{proof}[1][Proof]%
{\par\noindent\textit{#1:\ }}%
{\hspace*{\fill} \rule{6pt}{6pt}}
\newenvironment{proof*}[1][Proof]%
{\par\noindent\textit{#1:\ }}{}

{\setlength{\arraycolsep}{0.5mm}\left\{ \; \begin{array}{#1}}%
    {\end{array} \right.}
\newenvironment{system*}[1]%
{\setlength{\arraycolsep}{0.5mm} \begin{array}{#1}}%
  {\end{array}}

\newtheorem{lemma}{\textbf{Lemma}}
\newtheorem{definition}{\textbf{Definition}}
\newtheorem{theorem}{\textbf{Theorem}}

\newtheorem{remark}{\textbf{Remark}}
\newtheorem{example}{\textbf{Example}}

\begin{document}

\def\figurename{Fig.}

\title{On MIMO Stability Analysis Methods Applied to Inverter‑Based Resources
Connected to Power Systems}    
\author{Anton A. Stoorvogel, \IEEEmembership{Senior Member, IEEE}, and
  Saeed Lotfifard, \IEEEmembership{Senior Member, IEEE} Ali Saberi,
	\IEEEmembership{Life Fellow, IEEE}, 
\thanks{Anton A. Stoorvogel is with the Department of Electrical
	Engineering, Mathematics and Computer Science, University of
	Twente, P.O. Box 217, Enschede, The Netherlands (e-mail:
	A.A.Stoorvogel@utwente.nl)}
\thanks{Saeed Lotfifard and Ali Saberi are with the
	School of Electrical Engineering and Computer Science, Washington
	State University, Pullman, WA 99164, USA (e-mail: s.lotfifard@wsu.edu,saberi@wsu.edu)}
}
\maketitle
\begin{abstract}
  This paper presents a critical review of methods
  commonly employed in the literature for small signal stability analysis of
  inverter based resources (IBRs). It discusses the intended purposes
  of these methods and outlines both their proper and improper
  implementations. The paper provides insights into the applicability
  of these techniques, clarifies their inherent limitations, and
  discusses and illustrates common sources of misinterpretation.   
\end{abstract}


\section{Introduction}

Numerous methods have been proposed for small-signal stability
analysis of inverter-base resources (IBRs) in the recent years. A
significant number of papers have focused on applying frequency-domain
stability analysis to multi-input multi-output (MIMO) systems. While
mainstream control theory has largely moved away from classical
frequency-domain tools, such as Nyquist plots or gain and phase
margins, for MIMO systems, many recent studies on IBRs revisit these
tools in their analyses. A review of this body of work, comprising
more than one hundred papers, reveals some common mistakes:
\begin{itemize}
\item failing to clearly and formally define the objective
  of the analysis 
\item selecting inappropriate MIMO stability analysis techniques for
  the specific application,
\item incorrect implementation or misapplication of the chosen
  methods,
\end{itemize}
This paper provides an in-depth discussion of established MIMO
stability methods and explains their proper use in this
context. Regarding the first point, the papers often state that their
main objective is to check the stability of a closed-loop linear
system with a known transfer matrix. However, this is trivial to
check. It seems that many of these papers have secondary aims. The
papers check properties like gain and phase margin. But this is not
needed if the sole purpose is to check stability so one must assume
that the authors want to guarantee stability given model inaccuracies
but this is not clarified.

The rest of this paper is organized as follows. Section \ref{II}
presents a brief discussion of SISO stability analysis, which provides
background and conceptual grounding for the MIMO methods discussed in
later sections. Section \ref{IIIa}, \ref{IIIb}, \ref{IIIc}, \ref{IIId}
and \ref{IIIe} discuss different aspects of MIMO stability
analysis. Finally, conclusions are presented in section \ref{IV},

\section{Single-input, single-output (SISO) systems}\label{II}

In many papers on IBRs stability analysis such as
\cite{rosso-andresen-engelken-liserre,sun,wen-burgos-boroyevich-mattavelli-shen},
the stated purpose of the paper is to verify the stability of a
particular system often given in the form of a feedback loop as in
Fig.\ref{pic0}
\begin{figure}[t]
  \centering
  \begin{tikzpicture}[every node/.style={outer sep=0pt,thick},
    decorate, scale=0.4]
    \draw (4,6) rectangle (8,10) node[pos=.5] {\large $P$};
    \draw[very thick] (2,8) circle [radius=0.3];
    \node at (2.5,6.5)  {\tiny $-$};
    \node at (1,8.5)  {\tiny $+$};    
    \draw[-{Latex[length=2mm]}]  (8,8) -- (14,8);
    \draw[-{Latex[length=2mm]}]  (0,8) -- (1.7,8);
    \draw[-{Latex[length=2mm]}]  (2.3,8) -- (4,8);
    \draw  (11,8) -- (11,5);
    \draw  (2,5) -- (11,5);
    \draw[-{Latex[length=2mm]}]  (2,5) -- (2,7.7);
  \end{tikzpicture}
  \caption{Standard feedback loop}\label{pic0}
\end{figure}
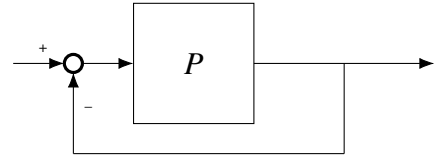
where $P$ is a linear system. In the case of single-input,
single-output (SISO) systems the system can be described by a transfer
function $P(s)$ and the transfer function of the intersection is then
given by
\[
  \frac{P(s)}{1+P(s)}=1-\frac{1}{1+P(s)}
\]
If
\[
  P(s)=\frac{n(s)}{d(s)}
\]
with $n$ and $d$ the numerator and denominator polynomials then
stability of this interconnection basically amounts to checking that
the polynomial $n(s)+d(s)$ has all zeros in the open left half plane
which is easily verified by, for instance, the Routh-Hurwitz test, see
\cite{meinsma3}. However, in the IBRs stability analysis literature
commnly use the Nyquist criterion or the Bode plot. In the Nyquist
criterion one examines whether the closed curve $P(j\omega)$ with
$\omega\in \R$ with increasing $\omega$ encircles the point $-1$. The
conclusion, that the Nyquist curve does not encircle or intersect with
the point $-1$ and hence the system is stable is often made. However,
this is only true if the system $P(s)$ is stable. In general the
number of counterclockwise encirclements should be equal to the number
of unstable poles of the system $P(s)$ (and checking the stability of
$P(s)$ is as much work as checking the stability of the closed loop
system). If $P(s)$ has poles on the imaginary axis, one has to be a
bit more careful as clearly explained in for instance
\cite{franklin-powell-emami-naeini}. Some methods only check whether
the Nyquist curve does not intersect the the line segment
$(-\infty,-1)$ (for instance easy to check using the Bode
plot). However, this needs some words of caution:
\begin{itemize}
\item This is not a necessary condition. There are systems which
  yields a stable closed loop system even though the Nyquist curve
  intersects the line segment $(-\infty,-1)$ on the real axis. An
  example of this phenomenon is given by:
  \[
    P(s)=\tfrac{5(s+10)^2}{(s+0.3)^3}
  \]
\item One has to be careful what is meant by intersection. The system
  \[
    P(s)=\tfrac{1-2s}{1+s}
  \]
  yields an unstable closed loop system but the Nyquist curve only
  intersects $(-\infty,-1)$ in the limit, i.e. the following holds:
  \[
    \lim_{\omega\rightarrow \infty} P(j\omega)=-2
  \]
  is on the interval  $(-\infty,-1)$.
\end{itemize}
Also there are papers that only plot the Nyquist curve for positive
frequencies which is a bad idea because one has no closed curves
making it hard to see the encirclements.

Given all the above, why do many studies on inverter‑based resource
(IBR) stability analysis employ the Nyquist criterion instead of just
checking the stability of the closed-loop transfer function? Many
papers, do not explicitly give a reason but stability of a system is
clearly only one condition that want to impose on a closed loop
system. Disturbances may affect the system. It is desirable to have a
system with limited overshoot and quick response to command
signals. The model might not be perfect due to uncertain parameters,
ignoring fast dynamics or nonlinearities and the like. For SISO
systems, the Bode plot (amplitude and phase) completely describes the
system. The Nyquist curve almost completely describes the system (one
cannot determine bandwidth from the Nyquist contour).  By checking
properties like gain and phase margin (which will be discussed in more
detail in the next section) but also high frequency rolloff and the
like one can easily analyse the properties of the system. These
properties are easily checked using the Nyquist curve or the Bode
plot.  One can generally say that a decent gain and phase margin
implies that the system behaves decently. It can handle model
uncertainty, it does not have a huge overshoot and the like. There are
some exceptions (often mathematically constructed) and clearly more
analysis can be done but this is a decent rule of thumb. Hence for
SISO systems Nyquist curve is commonly used to check gain/phase margin
and analyse the Bode plot to investigate the frequency behavior. More
importantly it also is a beautiful tool for controller design through
loopshaping. For SISO systems this frequency domain analysis is still
standard and is taught in almost all undergraduate control classes.

\section{Multiple-input, multiple-output (MIMO) systems}\label{IIIa}

In the 1980s control theory moved to multi-input, multi-output systems
and controller design. The advances in computing made this possible
and it was required to achieve objectives in many
applications. Therefore, there was an urgency felt to expand
loopshaping and other frequency-domain techniques to the multi-input,
multi-output case. This was not easy. For instance, the concept of
zeros becomes much more difficult because of the directionaliy issues
introduced in the MIMO context. Many other issues arose, some of which
outlined below, but essentially the conclusion in the area of control
was that to properly handle MIMO systems one has to move to state
space models and use different tools. For instance, gain and phase
margins can not be extended to the MIMO context in a reasonable and
effective way as will be explained later. Effectively, since around
year 2000, the mainstream control community has abandoned studying
MIMO systems in the frequency-domain using tools like Nyquist or
gain/phase margin. This was intially resisted by many because it
required some new specific mathematical skills not present in every
control engineer.

Now that papers on IBRs stability analysis are moving toward the MIMO
stability analysis such as
\cite{ren-duan-chen-huang-wu-min,yu-lin-song-yu-yang-su,cao-ma-yang-wang-tolbert,
  samanes-urtasun-barrios-lumbreras-lopez-gubia-sanchis,liao-wang,jiang-konstantinou}
one can see that they also focus on frequency domain techniques
referring to control papers and books from before 2000. The following
sections outline some of the issues encountered in the attempts to use
frequency domain techniques for MIMO stability analysis of IBRs.

Note, like for SISO systems, if the sole objective is to check the stability
of an interconnection such as in \ref{pic0} then one can simply check
whether the location of the poles of
\begin{equation}\label{eq1aa}
  P(I+P)^{-1}=I-(I+P)^{-1}
\end{equation}
are in the open left half plane. In the context of MIMO systems, for
an $n\times n$ transfer matrix one would have to check each of the
$n^2$ transfer functions from each input  to each output but this is
still quite trivial unless dimensions really blow up. It may appear 
that the analysis can be simplified by looking at the determinant of
\begin{equation}\label{eq1a}
  \det (I+P)
\end{equation}
However, stability of the inverse of \eqref{eq1a} does not always
guarantee stability of the closed loop system. 

\begin{theorem}\label{theo1}
  The closed loop system is stable \eqref{eq1aa} if and only if
  \begin{itemize}
  \item The determinant in \eqref{eq1a} has no unstable zeros.
  \item The number of unstable poles of the determinant in
    \eqref{eq1a} should be equal to the number of unstable poles of
    $P$, counting multiplicities.
  \end{itemize}
\end{theorem}

In this definition, one has to be careful to count the number of
unstable poles of $P$. This has to be done based on the Smith-McMillan
form presented in \cite{ro}. Clearly if $P$ is stable, one does not have
to worry about this. It is sometimes stated
that the system should have no hidden unstable modes (without a clear
definition). In the above, a formal definition is given by what it
means to not have hidden unstable modes: the multiplicity of the unstable poles
in $P$ should be equal to the multiplicity of the unstable poles
in the determinant of $I+P$.

The fact that determining the multiplicity is not as trivial as one
might think is illustrated via an example.

\begin{example}
  Consider the following two transfer matrices:
  \[
    P_1(s)=\begin{pmatrix}
      \frac{1}{s+1} & \frac{1}{s+2} & \frac{1}{s-1} \\[2mm]
      0             & \frac{2}{s+2} & \frac{1}{s+1} \\[2mm]
      0             & 0             & \frac{s+2}{s-1}
    \end{pmatrix},\ 
    P_2(s)=\begin{pmatrix}
            \frac{1}{s+1} & \frac{1}{s-1} & \frac{1}{s+2} \\[2mm]
      0             & \frac{2}{s+2} & \frac{1}{s+1} \\[2mm]
      0             & 0             & \frac{s+2}{s-1}
    \end{pmatrix}.
  \]
  The transfer matrices both have an unstable pole in $1$. Using the
  Smith-McMillan form it can be shown that the unstable pole $1$ of
  $P_1$ has multiplicity equal to $1$ while the unstable pole $1$ of
  $P_2$ has multiplicity equal to $2$. On the other hand:
  \[
    \det(I+P_1)=\det(I+P_2)= \frac{2s+4}{(s+1)(s+2)(s-1)}
  \]
  so for $P_1$ the multiplicity of the unstable pole in is equal to
  the multiplicity of the unstable pole of the determinant (no hidden
  mode) while for $P_2$ this is not the case (hidden mode).
\end{example}

\section{Generalized Nyquist}\label{IIIb}

For SISO systems, a single Nyquist curve or a Bode plot completely
describes the system. In contrast, for $n\times n$ MIMO systems, there
are essentially $n^2$ transfer functions and hence there are $n^2$
possible Nyquist curves or Bode plots that one could investigate. On
the other hand, it is obvious that while looking at $4$ Bode plots for
a $2\times 2$ system it is really hard to fully grasp the systems
behavior.

Therefore the $\det(I+P)$ is considered. As stated before (with some
limitations) one can assess the stability of the system based on the
determinant. However, one looses all directionality information and it
is basically impossible to deduce anything regarding the behavior of
the system by only looking at this determinant.

\begin{example}
  Investigating only $\det(I+P)$ has intrinsic limitations because it
  does not capture many aspects of a MIMO system. 
  Consider the following systems:
  \begin{alignat*}{3}
    & \begin{pmatrix} \tfrac{1}{s+1} & 0 \\ 0 &
      \tfrac{1}{s+1} \end{pmatrix},\ && \
    \begin{pmatrix} 1 & 0 \\ 0 &
      \tfrac{1}{(s+1)^2} \end{pmatrix},\\
    & \begin{pmatrix} \tfrac{1}{s+1} & \tfrac{1}{(s+2)(s+0.1)}\\ 0 &
      \tfrac{1}{s+1} \end{pmatrix},\ && \ 
    \begin{pmatrix} \tfrac{2}{s+1} & \tfrac{2s+3}{s+1} \\ \tfrac{1}{s+1} &
      \tfrac{1}{s+1} \end{pmatrix}.
  \end{alignat*}
  Even though each of these systems has an intrinsically different
  behavior, they all produce the same $\det(I+P)$. This illustrates
  that many properties (good or bad) of a system are not visible if
  one only analyzes the determinant.
\end{example}

If one looks at stability tests for MIMO systems, then checking
whether $\det(I+P)$ has unstable zeros can be evaluated using a
version of the Nyquist criterion. If the Nyquist curve of $\det(I+P)$
is obtained, then one should verify whether the number of
counterclockwise encirclements of the point $0$ equals the number of
unstable poles of $P$ (counting multiplicities). Note the fact that
this time one must check encircling $0$ not $1$. In principle correct
but one needs to check the multiplicity of unstable poles of $P$
carefully (using the Smith-McMillan form) and it is easier to simply
check the pole locations of $(I+P)^{-1}$ directly.

The reason why initially the Nyquist criterion was investigated for
MIMO systems was to extend concepts such as gain/phase margin and loop
shaping (which were extremely powerful for SISO systems) to the MIMO
domain.

If one thinks of expanding the gain margin to the MIMO domain the
logical first step is to look at $\det(I+kP)$ and check for which $k$
the system is stable. However, due to the determinant, this function
depends nonlinearly on the parameter $k$ and the gain margin can no
longer be obtained from the Nyquist curve of $\det(I+P)$. The proposed
solution was to look at the eigenvalues of $P(j\omega)$:
\[
  \lambda_1\left[ P(j\omega) \right],\ldots,
  \lambda_n\left[ P(j\omega) \right]
\]
After all, the following holds:
\[
  \det(I+P)=\prod_{i=1}^n (1+\lambda_i [P])
\]
This makes it possible to relate the encirclements of $0$ for the
determinant to encirclements of $-1$ for the eigenvalues. However, if
one wants to count encirclements of $-1$ then one needs closed
curves. There is one obvious reason and one very subtle reason why one
might not obtain closed curves.

The obvious reason is the way one numbers the eigenvalues. It is
needed that $\lambda_i \left[ P(j\omega) \right]$ are continous
functions of $\omega$ for each $i$. A simple ordering that $\lambda_1$
is the largest eigenvalue in amplitude (as done in some paper in the
power literature) does not work and can create
discontinuities. However, it is known that a continuous ordering
exists.

The second potential issue is illustrated in the following example.

\begin{example}\label{ex1}
  Consider the loop gain:
  \[
    P(s)=\begin{pmatrix} 0 & 1 \\ \tfrac{3(s-1)}{s+1} &
      -3 \end{pmatrix}
  \]
  The Nyquist plot based on $\det(I+P)$ is obtained in Fig.\ \ref{pic02}.
  \begin{figure}[t]
    \centering
    \resizebox{5cm}{!}{\includegraphics{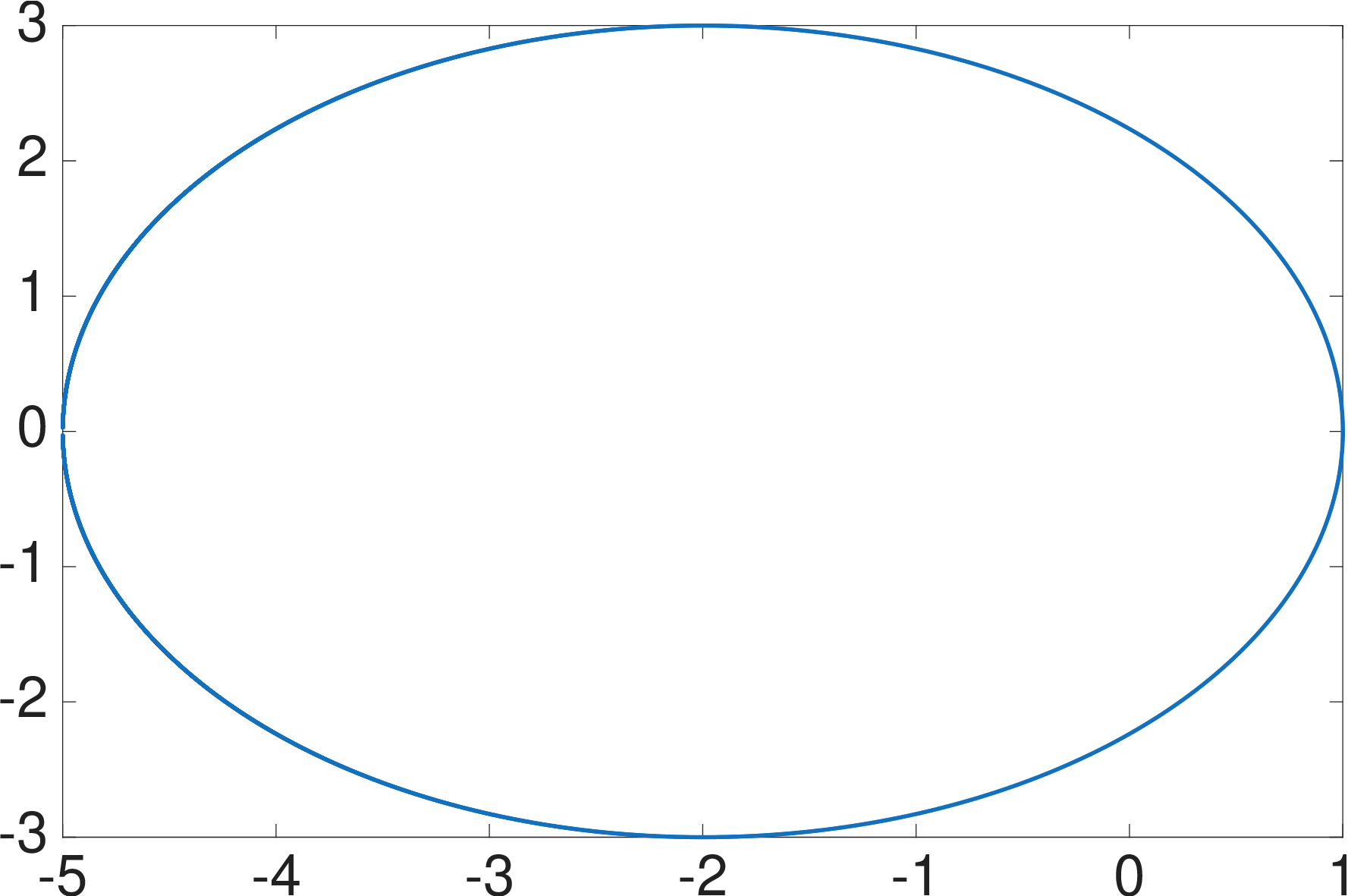}}
    \caption{Nyquist plot of $\det(I+P(j\omega))$ of Example \ref{ex1}}\label{pic02}
  \end{figure}
  
  One sees that the graph encircles $0$ while $P$ has no unstable
  poles and hence one concludes that the interconnection is
  unstable. This is of course trivially seen from the fact that:
  \[
    (I+P)^{-1} = \tfrac{s+1}{5s-1} \begin{pmatrix} 2 & 1 \\
      \tfrac{3(s-1)}{s+1} & -2 \end{pmatrix}
  \]
  which has an unstable pole in $\tfrac{1}{5}$. However, consider the 
  two plots of the two eigenvalues of the open loop gain in Fig.\ \ref{pic03}.
  \begin{figure}[t]
    \centering
    \resizebox{3.4cm}{!}{\includegraphics{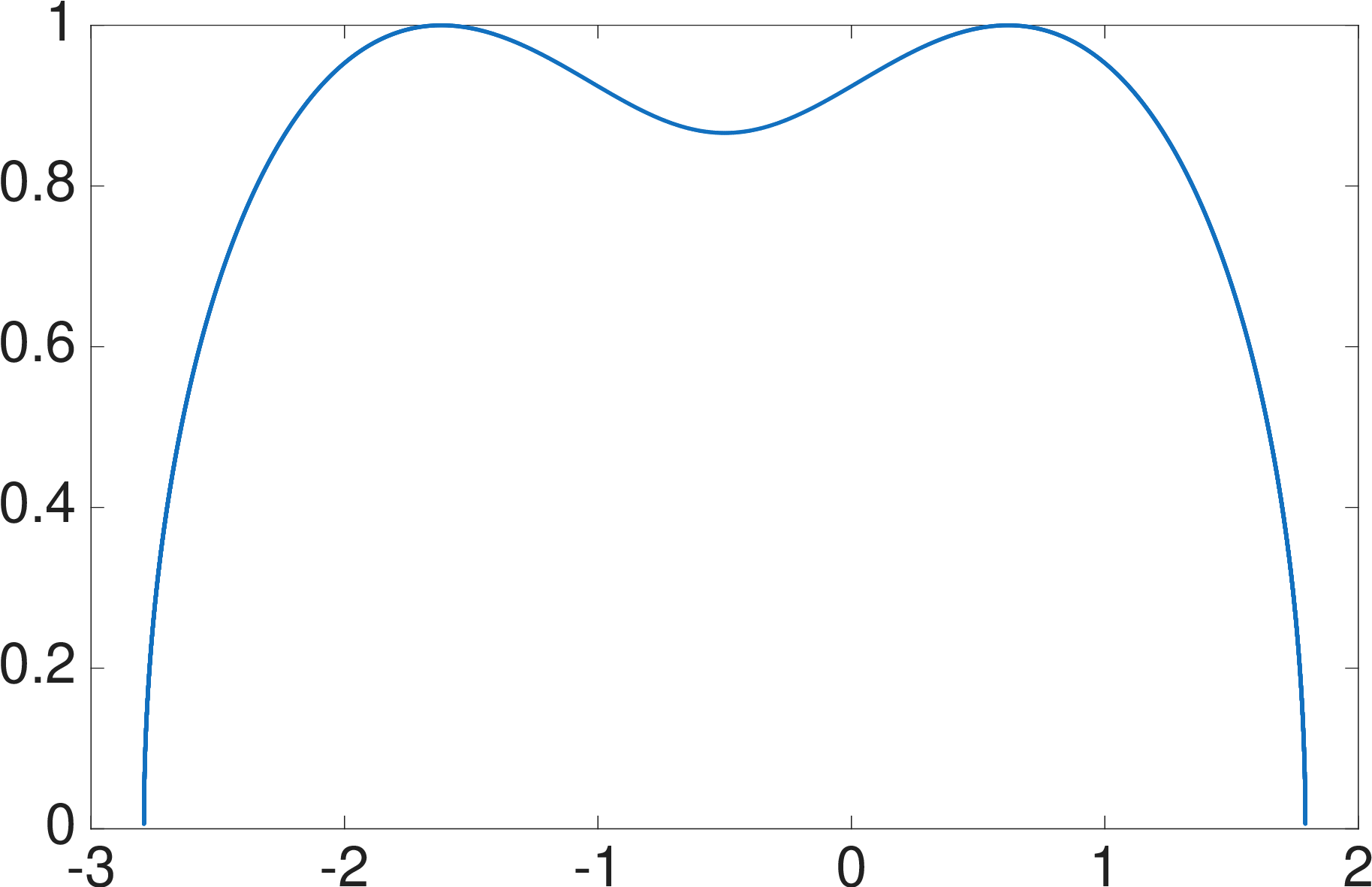}}
    \resizebox{3.4cm}{!}{\includegraphics{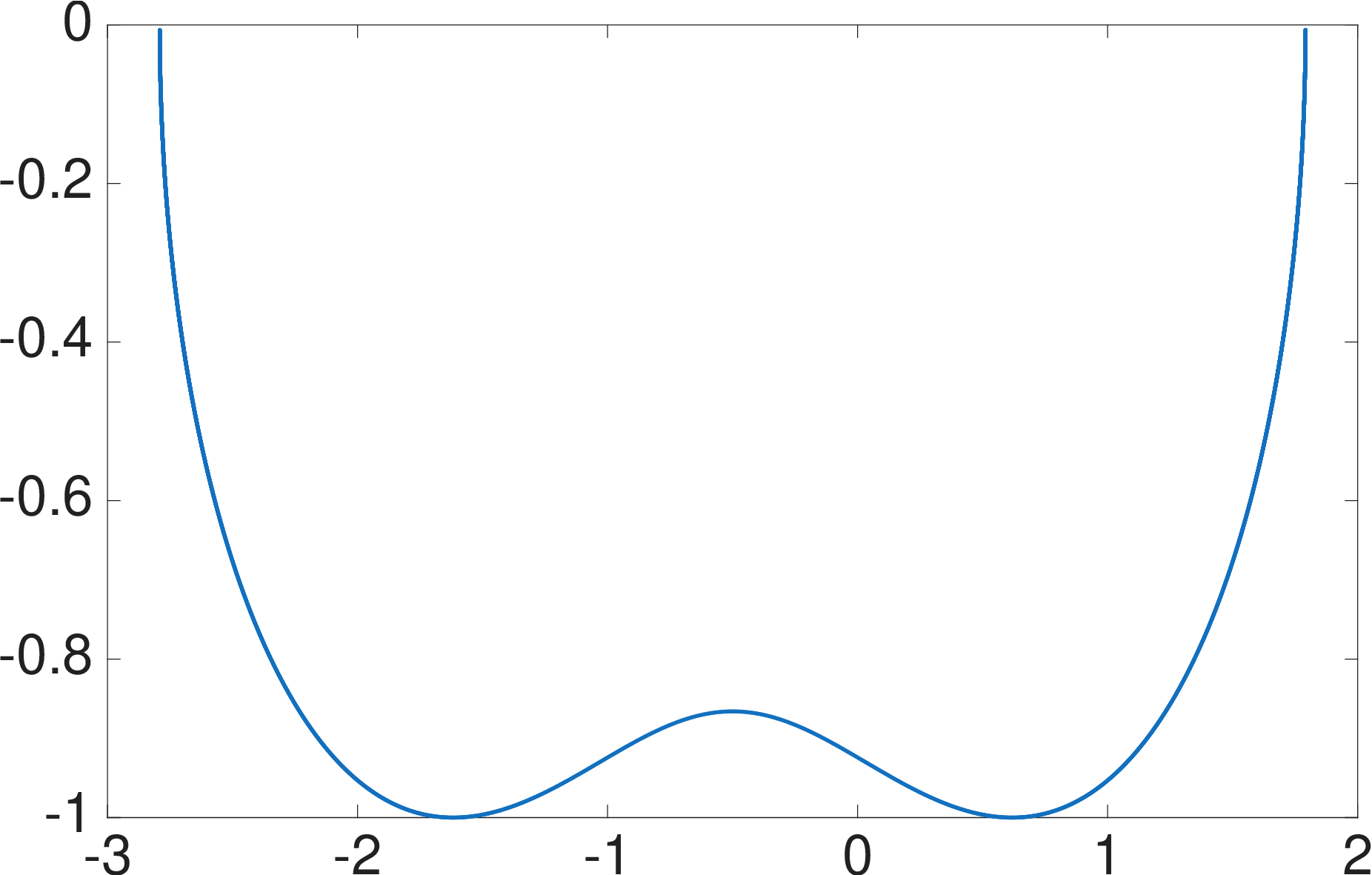}}
    \caption{Eigenvalue trajectories of $P(j\omega)$ of Example
      \ref{ex1} as the frequency varies from $-\infty$ to $\infty$}\label{pic03}
  \end{figure}

  Clearly, because these graphs are not closed one cannot even properly
  measure the number of encirclements.
\end{example}

This phenomenon was originally discovered by \cite{desoer-wang2} and
it was noted that the issue could be handled by pairwise combining
these $n$ eigenvalue graphs to create $q$ closed curves (with
$q\leq n$). Then, one should count the sum of the number of encirclements
in these $q$ graphs, and this should be equal to the number of
open-loop poles.

\begin{example}
  To continue example \ref{ex1}, two plots for the
  eigenvalues are obtained, neither of which forms a closed curve. Following
  \cite{desoer-wang2}, these curves are combined to obtain Fig.\ \ref{pic04}.
  \begin{figure}[t]
    \centering
    \resizebox{5cm}{!}{\includegraphics{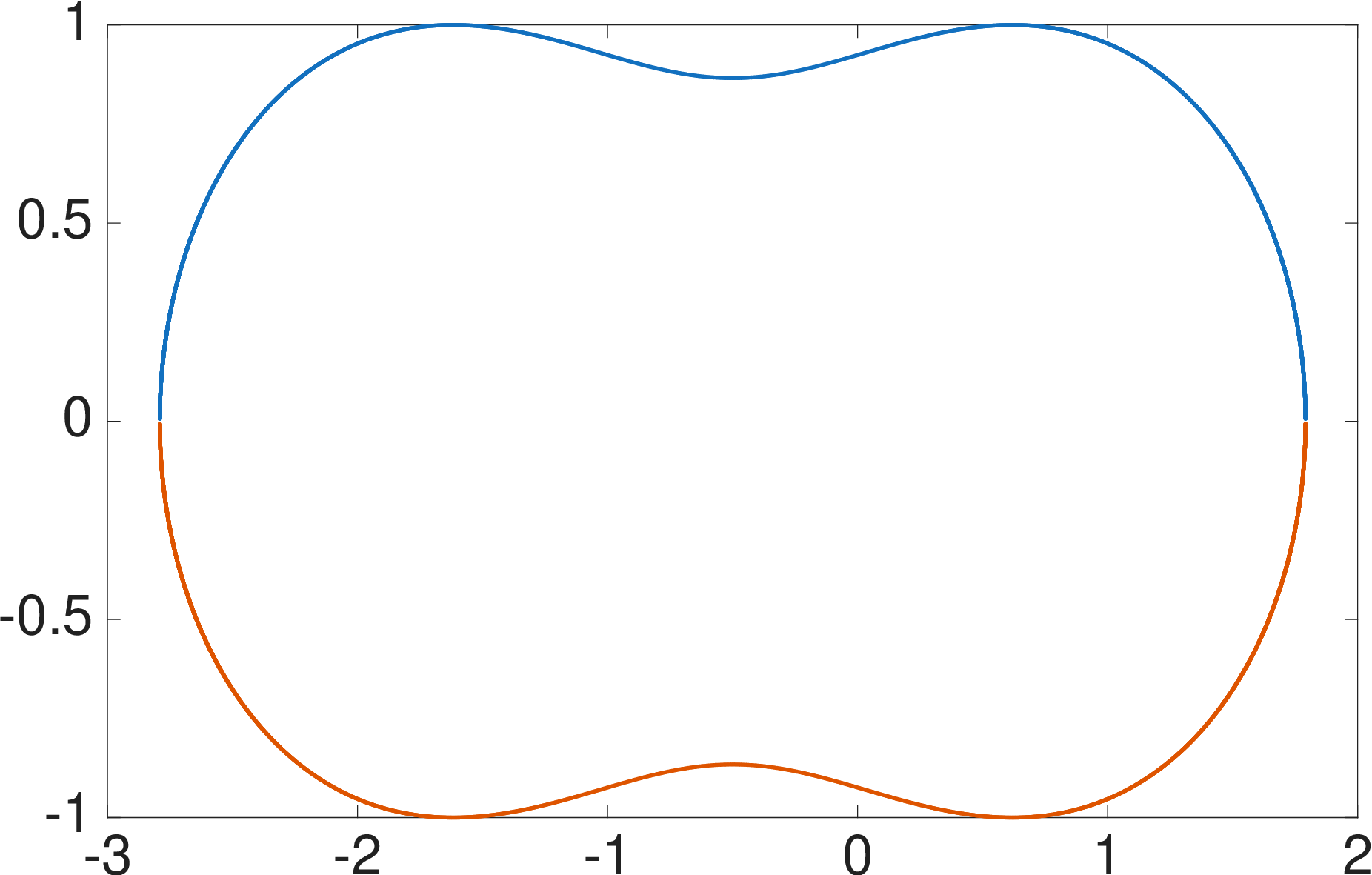}}
    \caption{Plot of combined trajectories from Fig.\ref{pic03}}\label{pic04}
  \end{figure}

  which encircles the point $-1$ once. Because the open‑loop system
  has no unstable poles, this implies that the closed-loop system is
  unstable. 
\end{example}

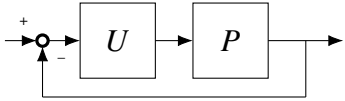
\begin{figure}[t]
  \centering
  \begin{tikzpicture}[every node/.style={outer sep=0pt,thick},
    decorate, scale=0.25]
    \draw (4,6) rectangle (8,10) node[pos=.5] {\large $U$};
    \draw (10,6) rectangle (14,10) node[pos=.5] {\large $P$};
    \draw[very thick] (2,8) circle [radius=0.3];
    \node at (3,7)  {\tiny $-$};
    \node at (1,9)  {\tiny $+$};    
    \draw[-{Latex[length=2mm]}]  (8,8) -- (10,8);
    \draw[-{Latex[length=2mm]}]  (0,8) -- (1.7,8);
    \draw[-{Latex[length=2mm]}]  (2.3,8) -- (4,8);
    \draw[-{Latex[length=2mm]}]  (14,8) -- (18,8);
    \draw  (16,8) -- (16,5);
    \draw  (2,5) -- (16,5);
    \draw[-{Latex[length=2mm]}]  (2,5) -- (2,7.7);
  \end{tikzpicture}
  \caption{Standard feedback loop with uncertainty block U}\label{pic1a}
\end{figure}
Gain and phase margin can be defined based on the feedback loop given
in Fig.\ \ref{pic1a}. For SISO systems, if one chooses $U=1$ then the
interconnection of Fig.\ \ref{pic0} is obtained. If this system is
stable then the question can be posed how sensitive this conclusion is
when small perturbations are introduced.  A logical step is to set
$U=k$ and check for which values of $k$ the closed loop system remains
stable. From the Nyquist curve one can easily obtain an interval
$[k_1,k_2]$ for $k$ with $k_1<1<k_2$ where stability is preserved. If
$k_1$ and $k_2$ are ``far'' away from $1$ then one can conclude that
stability is preserved given variations in the feedback gain. If one
chooses $U=e^{i\theta}$ then from the Nyquist curve one can easily
obtain an interval $[-\theta_1,\theta_1]$ for $\theta$ with $\theta_1>0$
where stability is preserved. If $\theta_1$ is ``far'' away from $0$
then one can conclude that stability is preserved given small delays
in the feedback gain (recall that delays yield a $e^{\tau s}$ term in
the frequency domain).

The question is whether the above can be used to obtain a form of
gain- and phase-margin for MIMO systems. Note that in Fig.
\ref{pic1a} the $U$ is no longer scalar but matrix valued. For $U=I$
closed loop interconnection of Fig.\ \ref{pic0} can be obtained.  Then the
gain margin can be defined as the largest interval $(k_1,k_2)$ with
$k_1<1<k_2$ such that the closed loop system is stable for all $U=kI$
with $k\in (k_1,k_2)$. Similarly, the phase margin is defined as the
largest $\theta_1>0$ such that the closed loop system is stable for
all $U=e^{j\theta} I$ with $\theta\in (-\theta_1,\theta_1)$.  This is
a direct extension of the SISO case. Since
\[
  kI=\begin{pmatrix} k & 0 & \cdots & 0 \\ 0 & k & \ddots & \vdots \\
    \vdots & \ddots & \ddots & 0 \\
    0 & \cdots & 0 & k \end{pmatrix}
\]
(and similarly for the phase margin) it can be observed that the
analysis assumes all channels have the same gain variation $k$. That
is why this is sometimes called uniform uncertainty.

The following result can be obtained:

\begin{theorem}[uniform uncertainty] \label{theo2} Assume the nominal
  closed loop system is stable ($U=I$). Consider the $n$ Nyquist plots
  obtained from the $n$ eigenvalues $\lambda_i\left[ P \right]$ for
  $i=1,\ldots, n$. Pairwise merge the Nyquist plots of eigenvalues
  which did not result in a closed curve to obtain $q$ closed Nyquist
  plots. For each of these individual plots compute $k_{1,i}$,
  $k_{2,i}$ and $\theta_{1,i}$ to obtain a gain-margin
  $[k_{1,i}, k_{2,i}]$ and a phase-margin $\theta_{1,i}$ based on the
  $i$th plot. In that case the gain- and phase margin for the MIMO
  system is given by:
  \[
    k_1 =\max_{i=\{1,\ldots, q\}} k_{1,i},\qquad
    k_2 =\min_{i=\{1,\ldots, q\}} k_{2,i},
  \]
  \[
    \theta_1 =\min_{i=\{1,\ldots, q\}} \theta_{1,i}
  \]
  In other words, the closed-loop system in Fig.\ \ref{pic1a} is
  asymptotically stable for any $U=kI$ with $k$ satisfying:
  \[
    k_1<k<k_2
  \]
  This gives the gain margin. Similarly, the closed-loop system is
  asymptotically stable for any $\theta$ satisfying
  \[
    -\theta_1 < \theta < \theta_1
  \]
  with $U=e^{j\theta}I$
\end{theorem}

However, using generalized Nyquist to obtain gain and phase margin
only verifies the effect of uniform changes in all channels. In most
cases, it makes no sense to only look at uniform changes in each
channel.  However, as noted below the Nyquist curve of the eigenvalues
misses essential information about the system.

To first natural step for a more reasonanle definition of gain and
phase margins in the MIMO case, is not to look at uniform disturbances
but look at potentially different perturbations in each channel:
\begin{equation}\label{nonuni}
  U = \begin{pmatrix}
    k_1    & 0      & \cdots & 0 \\
    0      & k_2    & \ddots & \vdots \\
    \vdots & \ddots & \ddots & 0 \\
    0      & \cdots & 0      & k_n
  \end{pmatrix}
\end{equation}
Here $k_i\in \R$ if one is looking for gain margin while
$k_i=e^{i\theta_i}$ if one is looking for phase margin.

However, it should first be noted that having satisfactory gain and
phase margins under uniform perturbations does not guarantee
robustness with respect to perturbations of the form \eqref{nonuni}. 

\begin{example}\label{ex3}
  Consider
  \[
    P(s)=\begin{pmatrix}
      \tfrac{1}{s+1}+\frac{b}{s+2} & \frac{b}{s+2}\\[2mm]
      -\frac{b}{s+2} & \tfrac{1}{s+1}-\frac{b}{s+2}
    \end{pmatrix}
  \]
  The gain margin can be computed and the closed-loop system is stable
  for $U=kI$ for any positive constant. Moreover, the closed-loop
  system is stable for $U=e^{j\theta} I$ provided
  $|\theta| < \tfrac{3\pi}{4}$. The parameter $b$ and the associated
  pole $s=-2$ does not effect the Nyquist curve.

  This system has a nice gain and phase margin considering uniform
  uncertainty. However, nonuniform perturbations strongly effects the
  behavior of the system. For instance, the closed-loop system becomes
  unstable if
  \[
    U=\begin{pmatrix} \tfrac{b-4}{b+4} & 0 \\[2mm] 0 & 1 \end{pmatrix}.
  \]
  which is clearly a very small deviation from $I$ for $b$ large. This
  example illustrates that the gain and phase margin computed through
  generalized Nyquist inherently only give information about
  identical perturbations in all channels and gives no insight in the
  effect of perturbations which are not identical for all channels.
\end{example}

Whether satisfactory gain and phase margins exist for nonuniform
perturbations can inherently not be checked using the generalized
Nyquist criterion. Another example illustrates that the issue is even
more involved since even if satisfactory phase margins exist for
nonuniform perturbations of the form \eqref{nonuni} (note that there
is no efficient tool to effectively compute gain and phase margins for
nonuniform perturbations) then one still cannot be sure that the system
is robust since nondiagonal perturbations might have dramatic effect
as illustrated by the following example.
\begin{figure}[t]
  \centering
  \begin{tikzpicture}[every node/.style={outer sep=0pt,thick},
    decorate, scale=0.25]
    \draw (4,6) rectangle (8,10) node[pos=.5] {\large $G_1$};
    \draw (10,6) rectangle (14,10) node[pos=.5] {\large $G_2$};
    \draw[very thick] (2,8) circle [radius=0.3];
    \node at (3,7)  {\tiny $-$};
    \node at (1,9)  {\tiny $+$};    
    \draw[-{Latex[length=2mm]}]  (8,8) -- (10,8);
    \draw[-{Latex[length=2mm]}]  (0,8) -- (1.7,8);
    \draw[-{Latex[length=2mm]}]  (2.3,8) -- (4,8);
    \draw[-{Latex[length=2mm]}]  (14,8) -- (18,8);
    \draw  (16,8) -- (16,5);
    \draw  (2,5) -- (16,5);
    \draw[-{Latex[length=2mm]}]  (2,5) -- (2,7.7);
  \end{tikzpicture}
  \caption{Standard interconnection of two stable systems $G_1$ and $G_2$ in
    small  gain theorem}\label{pic1a2} 
\end{figure}
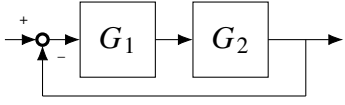

\begin{example}\label{ex2}
  Consider
  \[
    P(s)=\begin{pmatrix}
      \tfrac{1}{s+1} & \tfrac{b}{s+1} \\[2mm]
      0 & \tfrac{1}{s+1}
    \end{pmatrix}
  \]
  The gain margin can be computed and the closed-loop
  system is stable for
  \[
    K=\begin{pmatrix} k_1 & 0 \\ 0 & k_2 \end{pmatrix}
  \]
  for any positive constants $k_1,k_2$. Moreover, the closed-loop
  system is stable for
  \[
    K=\begin{pmatrix} e^{j\theta_1} & 0 \\ 0 &
      e^{j\theta_2} \end{pmatrix}
  \]
  provided $\theta_1,\theta_2 < \tfrac{3\pi}{4}$. However the $1,2$
  element of this matrix does not affect the Nyquist curves of the two
  eigenvalues in any way. On the other hand, this off-diagonal element
  strongly effects the behavior of the system. For instance, the
  closed-loop system becomes unstable if
  \[
    U=\begin{pmatrix} 1 & 0 \\ -\tfrac{5}{b} & 1 \end{pmatrix}.
  \]
  which is clearly a very small perturbation for $b$ large. This
  example, from \cite{lehtomaki-sandell-athans}, illustrates that it
  is dangerous to consider gain and phase margin as a good robustness
  measure for MIMO systems:
\end{example}

Note that some papers use tools such as Gershgorin circles to find
estimates for the eigenvalues and thus obtain simpler but often very
conservative methods to assess stability. This raises two
questions.
\begin{itemize}
\item Why introduce conservative tools to check stability if 
very efficient tools exist which are not conservative (actually
necessary and sufficient). 
\item These tools still only investigate the eigenvalues of $P$ and
  hence they miss intrinsic information about the system as is
  outlined in the above.
\end{itemize}

In conclusion, generalized Nyquist does not provide any insight on the
robustness of the closed loop stability. It intrinsically misses
information about the system to properly assess the behavior of the
closed loop system. If one is careful one can properly evaluate whether
the system is stable using generalized Nyquist. On the other hand, if
one wants to only check stability then one can much more efficiently
check the poles of the closed-loop system directly. There is no additional
insight obtained from generalized Nyquist and hence it has not been
used in the control area for the past 30 years.

\section{Small-gain theorem}\label{IIIc}

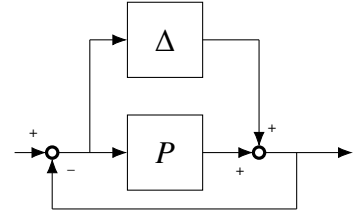
\begin{figure}[t]
  \centering
  \begin{tikzpicture}[every node/.style={outer sep=0pt,thick},
    decorate, scale=0.25]
    \draw (6,12) rectangle (10,16) node[pos=.5] {\large $\Delta$};
    \draw (6,6) rectangle (10,10) node[pos=.5] {\large $P$};
    \draw[very thick] (2,8) circle [radius=0.3];
    \node at (3,7)  {\tiny $-$};
    \node at (1,9)  {\tiny $+$};    
    \draw[very thick] (13,8) circle [radius=0.3];
    \node at (12,7)  {\tiny $+$};
    \node at (13.7,9.3)  {\tiny $+$};    
    \draw[-{Latex[length=2mm]}]  (4,14) -- (6,14);
    \draw[-{Latex[length=2mm]}]  (0,8) -- (1.7,8);
    \draw[-{Latex[length=2mm]}]  (2.3,8) -- (6,8);
    \draw[-{Latex[length=2mm]}]  (10,8) -- (12.7,8);
    \draw[-{Latex[length=2mm]}]  (13.3,8) -- (18,8);
    \draw[-{Latex[length=2mm]}]  (13,14) -- (13,8.3);
    \draw  (4,8) -- (4,14);
    \draw  (10,14) -- (13,14);
    \draw  (15,8) -- (15,5);
    \draw  (2,5) -- (15,5);
    \draw[-{Latex[length=2mm]}]  (2,5) -- (2,7.7);
  \end{tikzpicture}
  \caption{Additive uncertainty}\label{pic1b1}
\end{figure}
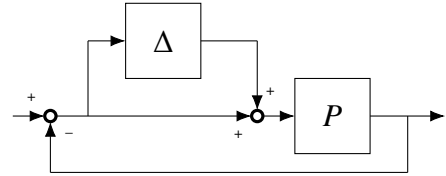
\begin{figure}[t]
  \centering
  \begin{tikzpicture}[every node/.style={outer sep=0pt,thick},
    decorate, scale=0.25]
    \draw (6,10) rectangle (10,14) node[pos=.5] {\large $\Delta$};
    \draw (15,6) rectangle (19,10) node[pos=.5] {\large $P$};
    \draw[very thick] (2,8) circle [radius=0.3];
    \node at (3,7)  {\tiny $-$};
    \node at (1,9)  {\tiny $+$};    
    \draw[very thick] (13,8) circle [radius=0.3];
    \node at (12,7)  {\tiny $+$};
    \node at (13.7,9.3)  {\tiny $+$};    
    \draw[-{Latex[length=2mm]}]  (4,12) -- (6,12);
    \draw[-{Latex[length=2mm]}]  (0,8) -- (1.7,8);
    \draw[-{Latex[length=2mm]}]  (2.3,8) -- (12.7,8);
    \draw[-{Latex[length=2mm]}]  (13.3,8) -- (15,8);
    \draw[-{Latex[length=2mm]}]  (19,8) -- (23,8);
    \draw[-{Latex[length=2mm]}]  (13,12) -- (13,8.3);
    \draw  (4,8) -- (4,12);
    \draw  (10,12) -- (13,12);
    \draw  (21,8) -- (21,5);
    \draw  (2,5) -- (21,5);
    \draw[-{Latex[length=2mm]}]  (2,5) -- (2,7.7);
  \end{tikzpicture}
  \caption{Multiplicative uncertainty}\label{pic1b2}
\end{figure}
As noted in the previous section, generalized Nyquist does not give
insights in robustness of the closed loop stability. Therefore,
different tools are needed. As noted in the literature the use of
eigenvalues is inherently insuitable in this context. For instance,
the classical reference for nonuniform perturbations as discussed in
the previous section is the paper \cite{lehtomaki-sandell-athans}
which uses singular values instead of eigenvalues. The basic tool for
many results in this context is the small gain theorem.

The $H_\infty$ norm is defined as follows:
\[
  \| P \|_\infty = \sup_{\omega} \| P(j\omega) \|.
\]
for a stable transfer function $P$ where $\| A \|$ for a matrix $A$ is
equal to the largest singular value.  An important property of the
$H_\infty$ norm is that
\[
 \| P_1P_2 \|_\infty \leq \| P_1 \|_\infty \| P_2 \|_\infty
\]
i.e. the norm is submultiplicative. The small gain theorem basically
states that the interconnection in Fig.\ \ref{pic1a2} is stable for
two stable systems $G_1$ and $G_2$ if $\| G_1 G_2 \|_\infty <1$.

The two classical models for uncertainty are additive and
multiplicative uncertainty as presented in Fig.\ \ref{pic1b1} and
\ref{pic1b2} respectively. They can be viewed as measuring absolute
and relative errors in the model.  The uncertainty structure presented
in the previous section basically sets $\Delta=U-I$ and using the
multiplicative model of Fig.\ \ref{pic1b2}. The small gain theorem basically
yields the following:

\begin{theorem}
  Assume that $P$ is stable. Either one of the interconnections in
  Fig.\ \ref{pic1b1} and \ref{pic1b2} is stable for $\Delta=0$ (nominal stability) if
  and only if $(I+P)^{-1}$ is asymptotically stable

  Given nominal stability, the intersection of Fig.
  \ref{pic1b1} (additive uncertainty) is stable for all
  $\Delta$ such that
  \[
    \| \Delta \|_\infty < \| (I+P)^{-1} \|^{-1}_\infty
  \]
  Given nominal stability, the intersection of Fig.\ \ref{pic1b2}
  (multiplicative uncertainty) is stable for all $\Delta$ such that
  \[
    \| \Delta \|_\infty < \| P(I+P)^{-1} \|^{-1}_\infty
  \]
\end{theorem} 

\begin{proof}
  The interconnection in Fig.\ \ref{pic1b1} can be rewritten in the
  form of Fig.\ \ref{pic1a2} by choosing $G_1=\Delta$ and $G_2=(I+P)^{-1}$.
  The small gain theorem then guarantees the interconnection is stable
  provided $\| G_1 G_2\|_{\infty}<1$. Since
  \begin{equation}\label{submulti}
    \| G_1 G_2 \|_{\infty}  \leq \| G_1 \|_{\infty} \| G_2
    \|_{\infty}
  \end{equation}
  The results for multiplicative uncertainty can be obtained using
  similar arguments using that the interconnection in Fig.\ \ref{pic1b2} can be rewritten in the
  form of Fig.\ \ref{pic1a2} by choosing $G_1=\Delta$ and
  $G_2=P(I+P)^{-1}$.
\end{proof}

For instance, if one considers the intersection in Fig.\ref{pic1a} in the
previous section with $U$ given by \eqref{nonuni} then \cite{lehtomaki-sandell-athans}
obtained estimates for the nonuniform gain and phasse margin based on the above
result with $\Delta=U-I$.

In the literature, different norms have been used. For instance, in
the IBRs stability analysis literature the $G$-norm is defined in \cite{belkhayat} which
does not satisfy the crucial submultiplicative property of
\eqref{submulti} and hence those results do not follow from a standard
application of the small gain theorem (which requires norm that are
submultiplicative). These results using the $G$-norm are actually
quite conservative.

However, besides the two models (multiplicative/additive uncertainty)
there is much more that can incorporated into the analysis:
\begin{itemize}
\item Sometimes, one has frequency information about the uncertainty such
  as:
  \[
    \| \Delta(j\omega) \| \leq f(j\omega)
  \]
  for $\omega\in \R$ with $f$ a known function can be incorporated into
  the analysis. This can, for instance, be used to incorporate
  information that there is more model uncertainty at high frequency.
\item There might be uncertain parameters in the model and the
  objective is to guarantee that the system remains stable for all
  parameter values in a certain range.
\item There might be inaccuracies in the model because of ignoring
  certain dynamics in the modeling. For instance, one might have a
  simple model $y_1=P_1u_1$ and $y_2=P_2u_2$ but might have ignored
  small interactions in the model meaning that $u_1$ does have
  some effect on $y_2$. These kinds of structured knowledge about the
  model uncertainty can be incorporated in the analysis.
\end{itemize}

Robust control can analyse the stability subject to all kind of model
uncertainty. The important issue is to model the uncertainty in the
sense that you incorporate into the analysis as much prior information
as one possibly can. Although much more could be discussed on this
topic, further details can be found in
\cite{chen-3,zhou-doyle-glover}.

Finally, it should be noted that in the context of IBRs stability
analysis, $P=YZ$ with admittance $Y$ and a grid modelled by an
infinite bus with impedance $Z$. In that case, one clearly can rewrite
Fig.\ \ref{pic0} into Fig.\ \ref{pic1a2} with $G_1=Z$ and $G_2=Y$.  Then,
in order to test for stability, one could use the small gain theorem to
conclude that the system is stable provided $\| YZ
\|_{\infty}<1$. This is a very conservative result but has been used
in the IBRs stability analysis literature with different types of
norms. Note that in the beginning of this section $G_1=\Delta$ was
used which is unknown. In the current context, both $G_1$ and $G_2$ are
known and one can just check the stability of the closed loop system
directly without invoking the small gain theorem. The small gain
theorem is conservative if you check for one specific $G_1$. But it
becomes necessary and sufficient if you want the check for stability
for all possible $G_1$ with norm less than $1$.

\section{Passivity}\label{IIId}

An intrinsically different approach to analyse the stability of closed
loop systems is to use the concept of passivity. This has also been
studied in the IBRs stability analysis literature, see \cite{
agbemuko-domínguez-garcía-gomis-bellmunt-harnefors,
harnefors-yepes-vidal-doval-gandoy,
zhao-wang-zhu,
harnefors-zhang-bongiorno}.

Passivity is based on the concept of energy. It is very easy to
conclude that a standard RLC circuit will always be stable because
this system does not contain any energy sources. Therefore, it is easy
to conclude that signals remain bounded unless external energy is
supplied and, because of resistors, energy dissipates. Therefore,
signals will converge to zero. This concept is not limited to linear
systems and these arguments equally apply to nonlinear circuits where
one, for instance, adds a diode or another nonlinear element. Also
passivity is used extensively in the context of robotics and other
mechanical structures were similar arguments can be applied.

First a formal definition and introduction to this concept is provided
which is based on
\cite{kottenstette-mccourt-xia-gupta-antsaklis,schaft}. Note that
these definitions apply equally well to SISO and MIMO systems.

\begin{definition}\label{def1}
  A (possibly nonlinear) system of the form:
  \[
    \begin{system*}{ccl}
      \dot{x}(t) &=& f(x(t),u(t)),\qquad x(0)=x_0 \\
      y(t) &=& h(x(t),u(t))
    \end{system*}
  \]
  where $f,h$ are locally Lipschitz with $f(0,0)=0, g(0,0)=0$ is
  called passive if there exists a function $\beta$ such that
  
  \[
    \int_0^T y\T(t)u(t)\, \textrm{d}t \geq  - \beta(x_0)
  \]
  for all input signals $u$ in $L_2$ and all initial conditions $x_0$
  and all $T>0$.
\end{definition}

\begin{figure}[t]
  \centering
  \begin{tikzpicture}[every node/.style={outer sep=0pt,thick},
    decorate, scale=0.25]
    \draw (4,6) rectangle (8,10) node[pos=.5] {\large $\Sigma_1$};
    \draw (10,6) rectangle (14,10) node[pos=.5] {\large $\Sigma_2$};
    \draw[very thick] (2,8) circle [radius=0.3];
    \node at (3,7)  {\tiny $-$};
    \node at (1,9)  {\tiny $+$};    
    \draw[-{Latex[length=2mm]}]  (8,8) -- (10,8);
    \draw[-{Latex[length=2mm]}]  (0,8) -- (1.7,8);
    \draw[-{Latex[length=2mm]}]  (2.3,8) -- (4,8);
    \draw[-{Latex[length=2mm]}]  (14,8) -- (18,8);
    \draw  (16,8) -- (16,5);
    \draw  (2,5) -- (16,5);
    \draw[-{Latex[length=2mm]}]  (2,5) -- (2,7.7);
  \end{tikzpicture}
  \caption{Standard interconnection of two systems $\Sigma_1$ and $\Sigma_2$ }\label{pic1a3}
\end{figure}
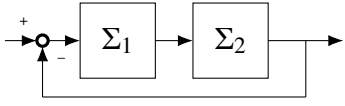

A standard misconception is that the interconnection of two passive
systems yields a stable interconnection. However, if in the
interconnection in Fig.\ \ref{pic1a3}, the, possibly nonlinear, systems
$\Sigma_1$ and $\Sigma_2$ are both passive then the interconnection
might still be unstable.

For that purpose the concepts of strict input- and output passivity
have been introduced:

\begin{definition}
  The system in Definition \ref{def1} is strictly-input passive if
  there exists a function $\beta$ and $\eps>0$ such that:
  \[
    \int_0^T y\T(t)u(t)\, \textrm{d}t \geq  \eps \int_0^T u\T(t)u(t)\,
    \textrm{d}t - \beta(x_0)
  \]
  for all input signals $u$ in $L_2$ and all initial conditions $x_0$
  and all $T>0$. Similarly, the system is called 
  strictly-output passive if there exists a function $\beta$ and
  $\delta>0$ such that:
  \[
    \int_0^T y\T(t)u(t)\, \textrm{d}t \geq  \delta \int_0^T y\T(t)y(t)\,
    \textrm{d}t - \beta(x_0)
  \]
  for all input signals $u$ in $L_2$ and for all initial conditions
  $x_0$ and all $T>0$.
\end{definition}

\begin{theorem}
  If in the interconnection in Fig.\ \ref{pic1a3} one of the systems
  is passive while the other system is strictly-input or -output
  passive then the interconnection is stable.
\end{theorem}

This result is very general and applies to nonlinear systems as well
as linear systems. The following simple example illustrates
that more than passivity of each system is needed to guarantee the stability
of the closed loop

\begin{example}
  Consider the interconnection in Fig.\ref{pic1a3} with $\Sigma_1$
  and $\Sigma_2$ are identical linear systems with transfer function:
  \[
    G(s)=\frac{6s^2+s+3}{s+2+3s+2}
  \]
  It can then be easily verified (using Theorem \ref{theorem1a}
  presented below) that $\Sigma_1$ and $\Sigma_2$ are both
  passive. However, it is easily checked that the interconnection in
  Fig.\ \ref{pic1a3} is not stable and has unstable poles in $j$ and
  $-j$.
\end{example}

For \textbf{linear} systems passivity can be connected to the positive
real property of the transfer system. Positive real has a very
different history and is intrinsically related to linear systems. If
one has a given transfer system $G$, one can ask whether it is
possible to construct an RLC circuit such that the transfer matrix of
this RLC circuit is equal to the given transfer matrix. The answer
turned out to be that this is possible when the transfer matrix is
positive real. In the literature there are different versions of
positive real being used:

\begin{definition}
  \[
    \begin{system*}{ccl}
      \dot{x}(t) &=& Ax(t)+Bu(t),\qquad x(0)=x_0 \\
      y(t) &=& Cx(t)+Du(t)
    \end{system*}
  \]
  with transfer matrix $G$. Assume the above is a minimal
  realization. In that case:
  \begin{itemize}
  \item The system is called positive real if
       \begin{enumerate}
       \item $G$ has no poles in the open right half plane.
       \item $G(j\omega)+G^{*}(j\omega) \geq 0$ for all $\omega$ such
         that $j\omega$ is not a pole of $G$.
       \item For any purely imaginary pole $j\omega_0$ of $G$ the following 
         \[
           \lim_{\omega\rightarrow \omega_0} (\omega-\omega_0) G(j\omega)
         \]
         exists and the limit is a positive semi-definite matrix.
       \end{enumerate}
  \item The system is called strictly positive real if there exists
    $\eps>0$ such that $G(s-\eps)$ is positive real.
  \item The system is called strongly positive real if there exists
    $\delta>0$ such that
    \begin{enumerate}
    \item $G$ has no poles in the closed right half plane.
    \item $G(j\omega)+G^{*}(j\omega) > \delta I$ for all $\omega\in
      \R$.
    \end{enumerate}
  \item The system is called ``strong'' positive real if
    \begin{enumerate}
    \item $G$ has no poles in the closed right half plane.
    \item $G(j\omega)+G^{*}(j\omega) > 0$ for all $\omega\in
      \R$.
    \end{enumerate}
  \end{itemize}
\end{definition}

\begin{remark}
  In the conditions characterizing strongly positive real, the second
  condition is sometimes presented as:
  \[
    G(j\omega)+G^{*}(j\omega) > 0 \text{ for all } \omega\in
    [-\infty,\infty]
  \]
  which is equivalent to our formulation.
\end{remark}

Let us relate these different concepts:
\begin{lemma}
  We have that
  \begin{multline*}
    \text{strong positive real } \Longrightarrow \text{ strictly
      positive real } \\
    \Longrightarrow \text{``strong'' positive real }
    \Longrightarrow \text{positive real }
  \end{multline*}
\end{lemma}

\begin{remark}
  None of these implications can be reversed in the above lemma.
  The system with transfer function:
  \[
    \frac{s+3}{(s+1)(s+2)}
  \]
  is strictly positive real but not strongly positive real while the
  system with transfer function:
  \[
    \frac{1}{s+1}
  \]
  is strictly positive real but not strongly positive real. Finally,
  the system with transfer function:
  \[
    \frac{6s^2+s+3}{s^2+3s+2}
  \]
  is positive real but not ``strong'' positive real.
\end{remark}

The following theorem connects the properties of positive real and
passivity for linear systems:

\begin{theorem}\label{theorem1a}
  Consider a linear system $\Sigma$ given by:
  \[
    \begin{system*}{ccl}
      \dot{x}(t) &=& Ax(t)+Bu(t),\qquad x(0)=x_0 \\
      y(t) &=& Cx(t)+Du(t)
    \end{system*}
  \]
  with transfer matrix $G$. Assume the above is a minimal
  realization. In that case
  \begin{itemize}
  \item The system $\Sigma$ is passive if and only if the transfer
    matrix $G$ is positive real.
  \item The system $\Sigma$ is strictly-input passive if and only if the transfer
    matrix $G$ is strongly positive real
  \end{itemize}
\end{theorem}

\begin{remark}
  Note that strictly-input passive implies strictly-output
  passive. However, a characterization of strict-output passive in
  terms of a version of positive real is quite involved and will not
  presented in this paper.
\end{remark}

Note that the following results exist regarding the stability of the
interconnection in Fig.\ \ref{pic1a3}.

\begin{theorem}\label{th6}
  Assume in the interconnection in Fig.\ \ref{pic1a3} both systems
  are linear time-invariant with a rational transfer matrix with one of
  the systems positive real wheree the other system is ``strong''
  positive real then the interconnection is stable.

  Assume in the interconnection in Fig.\ \ref{pic1a3} both systems
  are linear time-invariant with a possibly irrational transfer
  matrix where one of the systems is positive real while the other
  system is strictly positive real then the interconnection is stable.
\end{theorem}

\begin{remark}
  Note that irrational transfer matrix might arise due to, for
  instance, time delays.  
\end{remark}

The key strength of passivity is that one can arbitrarily connect
strict passive systems and preserve passivity and stability. For
instance, if one ensure a set of IBRs is strictly passive and they are
connected through an arbitrary network then passivity is preserved
assuming the network itself is passive. The latter is actually quite
natural since, for instance, a power line logically cannot be an
energy source.

\section{Mixed small gain theorem and positive real}\label{IIIe}

In the paper \cite{griggs-anderson-lanzon} the concepts of positive
real and the small-gain theorem are mixed which considers a
system with transfer matrix $G$ for which there exists a $c>0$ such
that:
\begin{itemize}
\item $G(j\omega)+G^{*}(j\omega) >0$ for $|\omega|\leq c$
\item $\sup_{\omega\in\R, |\omega|>c}  \| G(j\omega) \| < 1$ 
\end{itemize}

\begin{theorem}
  Consider the interconnection in Fig.\ref{pic1a3}. Assume both
  systems $\Sigma_1$ and $\Sigma_2$ are stable and have a transfer
  function satisfy the above mixed condition for the same $c$. In that
  case, the interconnection is stable.
\end{theorem}

The limitation of this theorem is that the components must satisfy these
conditions for the same bandwidth $c$.

If all inverters in a network are strictly passive then the
interconnection of all these inverters through a network is always
stable since obviously the network is  passive.

When one tries to achieve the same type of results here, it should be
noted that even if all the inverters satisfy this mixed property for
the same $c$, their interconnection through the network might still
yield an unstable system. After all, even though the network is
passive, this does not imply that it also satisfies such a mixed
property. 

\section{Conclusion}\label{IV}

This paper provided a critical review of the commonly used stability
analysis methods for IBRs. It presented authors' collective findings,
derived from a comprehensive examination of a large body of literature
on the stability of IBRs. Specifically, the paper clarified the
intended applications of the stability analysis methods and discussed
potential sources of misinterpretation or improper use. It underscored
the necessity of employing modern control theory for MIMO systems in
the stability analysis of IBRs. This necessity arises because
classical frequency‑domain techniques from the 1980s and 1990s,
developed as extensions of SISO approaches to MIMO systems, now
broadly regarded as inadequate for such analyses. Several of their key
limitations were discussed and illustrated in this paper.

\bibliographystyle{IEEEtran}
\bibliography{referenc}

\begin{IEEEbiography}[{\includegraphics[width=1.25in,height=1.5in,clip,keepaspectratio]
    {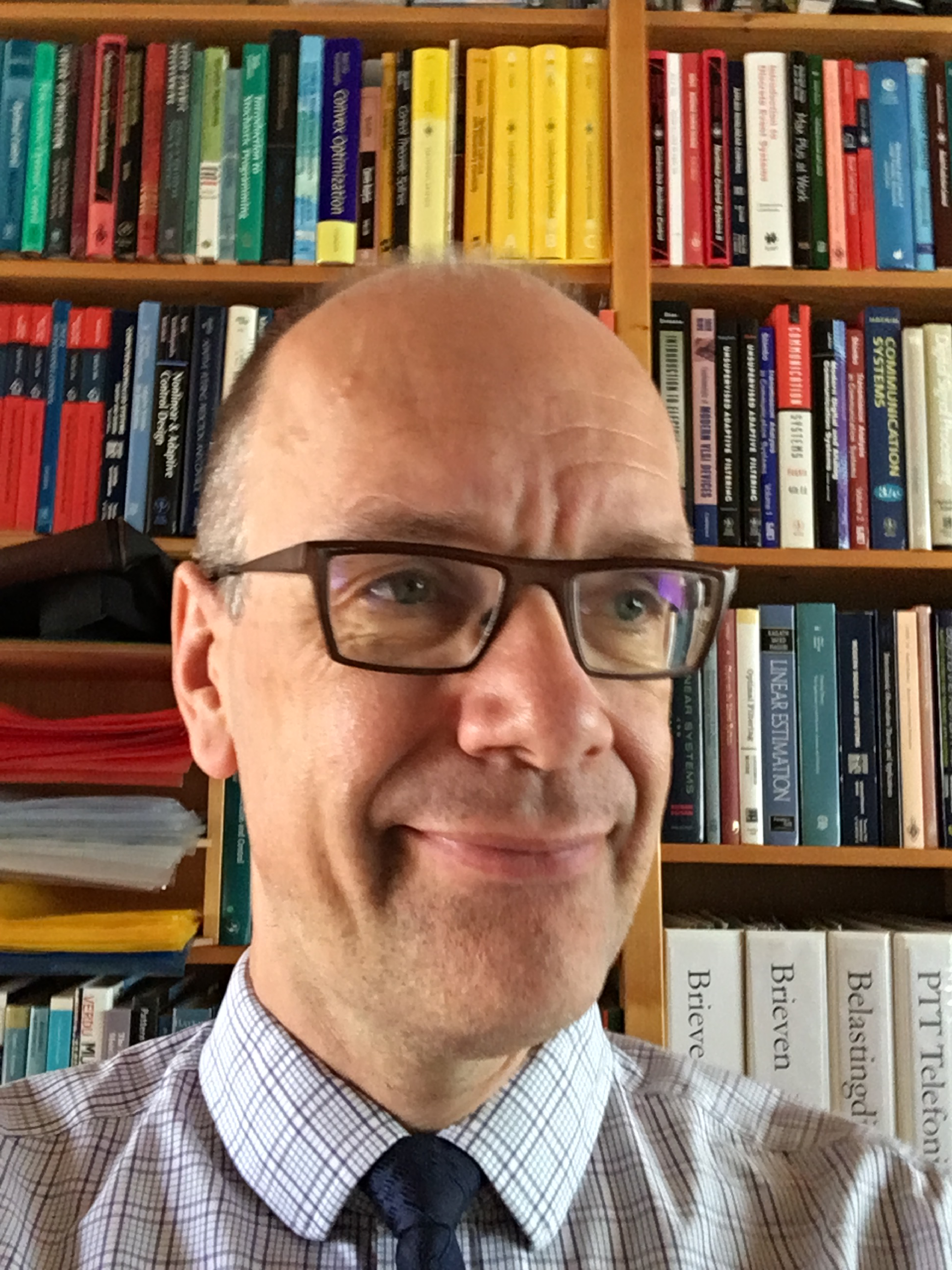}}]{Anton A. Stoorvogel}
  received the M.Sc. degree in Mathematics from Leiden University in
  1987 and the Ph.D. degree in Mathematics from Eindhoven University
  of Technology, the Netherlands in 1990. Currently, he is a professor
  in systems and control theory at the University of Twente, the
  Netherlands.  Anton Stoorvogel is the author of six books and
  numerous articles. He is and has been on the editorial board of
  several journals.
\end{IEEEbiography}%
\begin{IEEEbiography}[{\includegraphics[width=1in,height=1.25in,clip,keepaspectratio]
    {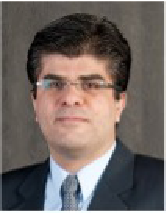}}]{Saeed Lotfifard}
  (S’08–M’11-SM’17) received his Ph.D. degree in
  electrical engineering from Texas A\&M University, College Station,
  TX, in 2011. Currently, he is an Associate Professor at Washington
  State University, Pullman. His research interests include stability,
  protection, and control of inverter based power
  systems. Dr. Lotfifard serves as an Editor for the IEEE Transaction
  on Power Delivery.
\end{IEEEbiography}%
\begin{IEEEbiography}[{\includegraphics[width=1.25in,height=1.5in,clip,keepaspectratio]
    {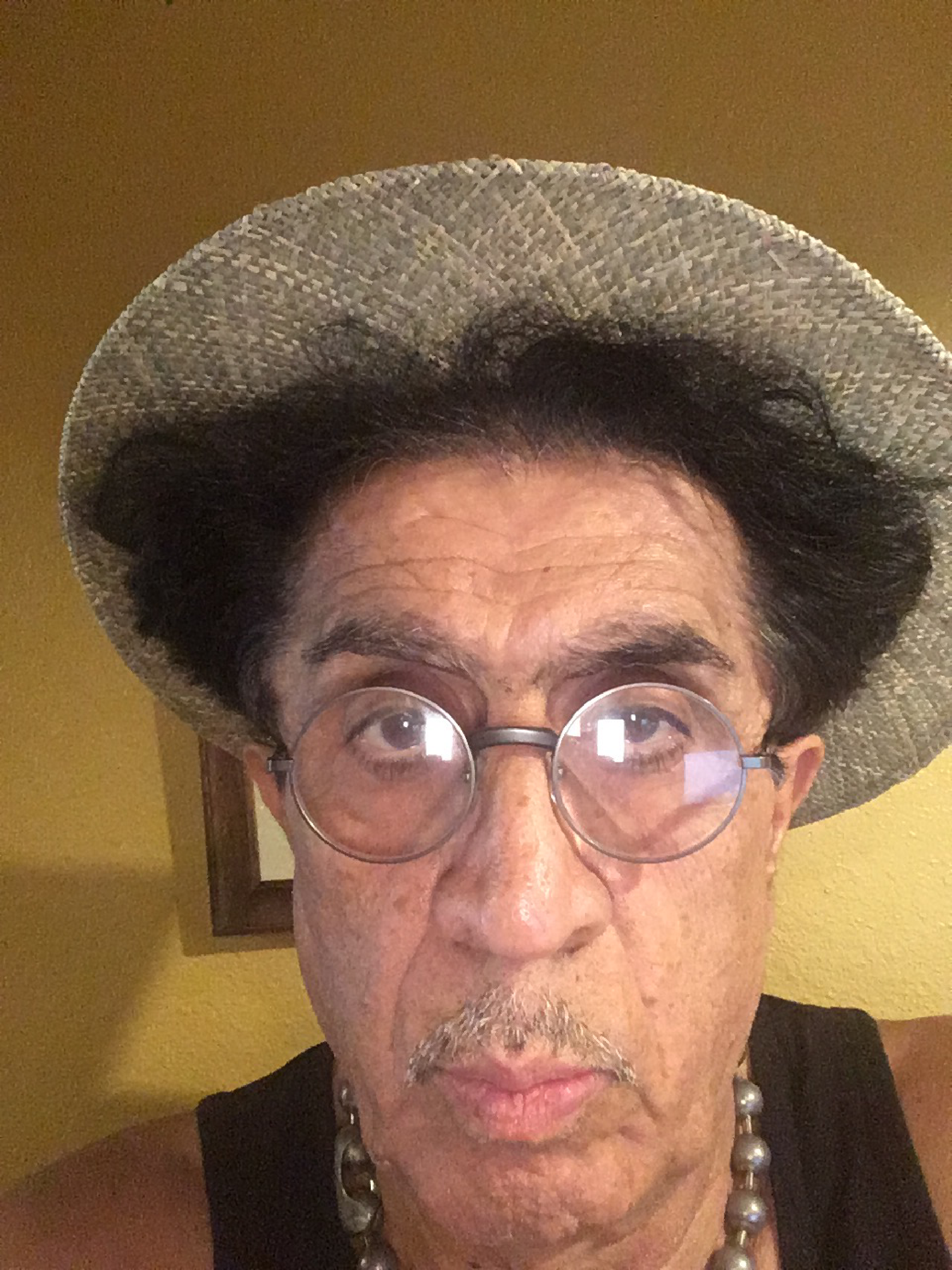}}]{Ali Saberi}
  lives and works in Pullman, Washington.
\end{IEEEbiography}
\vfill

\mbox{}

\end{document}